\documentclass[reqno, 12pt,a4letter]{amsart}
\usepackage{amsmath,amsxtra,amssymb,latexsym, amscd,amsthm,physics}
\usepackage[utf8]{inputenc}
\usepackage[mathscr]{euscript}
\usepackage{mathrsfs, mathtools}
\usepackage[english]{babel}
\usepackage{color}
\usepackage{float}
\usepackage{tikz}
\usetikzlibrary{patterns}

\theoremstyle{definition}

\newtheorem {remark}{Remark}[section]

\theoremstyle{plain}
\newtheorem {theorem}{Theorem}[section]

\newtheorem {conjecture}{Conjecture}[section]

\def\ees{{\accent"5E e}\kern-.385em\raise.2ex\hbox{\char'23}\kern-.08em}
\def\EES{{\accent"5E E}\kern-.5em\raise.8ex\hbox{\char'23 }}
\def\ow{o\kern-.42em\raise.82ex\hbox{
   \vrule width .12em height .0ex depth .075ex \kern-0.16em \char'56}\kern-.07em}
\def\OW{O\kern-.460em\raise1.36ex\hbox{
\vrule width .13em height .0ex depth .075ex \kern-0.16em
\char'56}\kern-.07em}

\title{On the maximum purity of absolutely separable bipartite states}

\author{Hoang Phi Dung$^1$, Vu The Khoi$^2$ }
\address{$^1$Department of Mathematics - Faculty of Fundamental Sciences,
	\newline \indent  Posts and Telecommunications Institute of Technology
	\newline \indent 
	Km10 Nguyen Trai Rd., Ha Dong District, Hanoi, Vietnam}
\email{dunghp@ptit.edu.vn}
\address{$^2$Institute of Mathematics, Vietnam Academy of Science and Technology, 	18 Hoang Quoc Viet, Ha Noi, Viet Nam}
\email{vtkhoi@math.ac.vn}

\date{ \today}

\begin{document}
\maketitle

\begin{abstract}
In this study, we investigate the problem of determining the maximum  purity for absolutely separable and absolutely PPT quantum states. From the geometric viewpoint, this problem is equivalent to asking for the exact Euclidean radius of the smallest ball around the maximally mixed state  that encompasses
the set of all absolutely separable or absolutely PPT states. Our results provide an analytic solution for two qubit states. Based on numerical computation, we propose a conjectured maximum purity for absolutely separable qubit-qudit states and absolutely PPT qutrit-qudit states.       

\end{abstract}

\section{Introduction}
Quantum entanglement is a fundamental phenomenon in quantum physics that plays a crucial role in both the theoretical and practical aspects of quantum information science. It describes a unique correlation between particles that remain interconnected, regardless of the distance separating them, allowing for the instantaneous transfer of information when the state of one particle is measured. Quantum entanglement serves a key resource in quantum communication and quantum computation \cite{Horodecki2009}.

In the context of quantum information theory, a state is said to be separable if it can be written as a convex combination of product states. Mathematically, a bipartite quantum state 
$\rho$ on a Hilbert space $\mathcal{H}=\mathcal{H}_A \otimes \mathcal{H}_B$, is called {\em separable} \cite{Schrodinger1935, Werner1989} if it can be  expressed as:
\begin{equation}\label{def-SEP}
	\rho = \sum_{i=1}^{k} p_i\rho_A^i \otimes {\rho}_B^i,	
\end{equation}
where $p_i\ge 0$ and $\sum_{i=1}^k p_i=1 ,$ and $\rho_A^i$  and ${\rho}_B^i$ are density matrices for systems $A$ and $B$ respectively. If $\rho$ is not a separable state, then $\rho$ is called an {\em entangled} state.

A state $\rho$ on $\mathcal{H}=\mathcal{H}_A \otimes \mathcal{H}_B$ is called {\it absolutely separable} \cite{Kus2001} if it remains separable under all possible global unitary transformations, that is $U\rho U^*$ is separable for all unitary operators $U: \mathcal{H}_A \otimes \mathcal{H}_B \rightarrow \mathcal{H}_A \otimes \mathcal{H}_B.$ Similarly, $\rho$ is termed {\it absolutely positive partial transpose} (absolutely PPT) if $U\rho U^*$ is PPT for all unitary operators $U: \mathcal{H}_A \otimes \mathcal{H}_B \rightarrow \mathcal{H}_A \otimes \mathcal{H}_B.$   

Absolutely separable states form a subset of all separable states and are of interest because they are never capable of becoming entangled through any quantum operation. This property makes absolute separability important in the study of quantum entanglement, as they represent the "least quantum" of all quantum states. Understanding these states can help clarify the boundaries between classical and quantum correlations.

Many studies have been devoted to the geometric structure of absolutely separable states. A central problem in the geometric study is to understand how this set fits within the entire state space and how ``large" it is compare to to the set of separable states. 
A notable result is that the largest ball of separable as well as absolutely separable states
around the maximally mixed state in a $n$-dimensional system has a radius of $\frac{1}{\sqrt{n(n-1)}}$ with respect to Hilbert-Schmidt metric \cite{Kus2001, Gurvits2002}. This ball is often referred to as the {\em maximal ball} and can also be described by the purity bound  ${\rm tr}(\rho^2)\le \frac{1}{n-1}.$ However, it was known that there are also absolutely separable states outside of this ball \cite{Kus2001}. This raises a natural question: what is the radius of the smallest ball around the maximally mixed state that encompasses all absolutely separable states? We can also rephrase this question as: what is the maximum purity of absolutely separable states?

Several authors have explored various aspects of this problem and related topics. The largest purity of absolutely separable two-qubit states was computed numerically in \cite{Slater2009}. An Upper bound for the purity of separable states with zero discord was established in \cite{Adhikari2021}, while  upper bounds for the purity and eigenvalues of absolutely PPT states was studied in \cite{Jivulescu2015,Filipop2017,Xiong2024}. 

In this paper, we derive the analytic maximum value of the purity for absolutely separable two qubit states. For higher dimensions, we numerically solve the optimization problem using MATLAB. Based on these numerical results, we formulate a conjecture about the maximum purity of absolutely separable qubit-qudit states and absolutely PPT qutrit-qudit states. 

The rest of this article is organized as follows: In the next section, we present known preliminary results regarding the characterization of absolutely separable and absolutely PPT states in terms of their spectrum. In Section 3, we find the maximum purity of absolutely separable two-qubit states and present our numerical results and conjectured maximum purity for qubit-qudit and qutrit-qudit states. Section 4 concludes.  
\section{Preliminaries}

Let $\rho$ be a mixed state on an $m\otimes n$ system $\mathcal{H}=\mathcal{H}_A \otimes \mathcal{H}_B,$ where $\dim \mathcal{H}_A=m$ and $\dim \mathcal{H}_B=n.$ Then, $\rho$ can be represented by a density matrix in $M_m\otimes M_n,$ where $M_n$ denotes the space of $n \times n$  matrices.  The {\em absolute separability problem} \cite{Kus2001, Hildebrand2007,Johnston2013, Arunachalam2015} (also called the separability from spectrum problem)  seeks to characterize and identify absolutely separable quantum states $\rho$   in terms of their spectrum. In the following, we state known criteria for absolutely separability and absolutely PPT  which will be used in the next section.    

The set of absolutely separable states is a subset of the set of absolutely PPT states and these two sets coincide in the case of the qubit-qudit system \cite{Johnston2013}. In this case, there is a simple criterion for determining whether a state is absolutely separable (absolutely PPT) based on its spectrum \cite{Hildebrand2007,Johnston2013, Arunachalam2015}. That is,  a state $\rho \in M_2 \otimes M_n$ is absolutely separable if and only if 
\begin{equation}\label{Eq2}
	\lambda_1 \leq \lambda_{2n-1} + 2\sqrt{\lambda_{2n-2}\lambda_{2n}},
\end{equation}
where $\lambda_1 \ge \lambda_2 \ge \cdots \ge \lambda_{2n}$ are the eigenvalues of $\rho$ in the decreasing order. 
Aside from the cases mentioned above, no known necessary and sufficient criterion exists for absolute separability.

For the qutrit-qudit system,  Hildebrand \cite[Corollary V.3]{Hildebrand2007} provided the following explicit criterion for absolute PPT: $\rho$ is absolutely PPT if and only if
{\small 
	\begin{equation} \label{Eq3}
		\begin{bmatrix}
			2\lambda_{3n} & \lambda_{3n-1} - \lambda_1 & \lambda_{3n-3} - \lambda_2\\
			\lambda_{3n-1} - \lambda_1 & 2\lambda_{3n-2} & \lambda_{3n-4} - \lambda_3\\
			\lambda_{3n-3} - \lambda_2 & \lambda_{3n-4} - \lambda_3 & 2\lambda_{3n-5}
		\end{bmatrix} \succeq 0,
		\begin{bmatrix}
			2\lambda_{3n} & \lambda_{3n-1} - \lambda_1 & \lambda_{3n-2} - \lambda_2\\
			\lambda_{3n-1} - \lambda_1 & 2\lambda_{3n-3} & \lambda_{3n-4} - \lambda_3\\
			\lambda_{3n-2} - \lambda_2 & \lambda_{3n-4} - \lambda_3 & 2\lambda_{3n-5}
		\end{bmatrix} \succeq 0, 
\end{equation} }
where $\lambda_1 \ge \lambda_2 \ge \cdots \ge \lambda_{3n}$ are eigenvalues of $\rho$ in the decreasing order and $\succeq 0$ indicates positive semidefiniteness.
It remains unknown whether the set of absolutely separable states coincides with the set of absolutely PPT states for the qutrit-qudit system \cite{Arunachalam2015}.  More recently, necessary and sufficient conditions conditions for absolute PPT of $4\otimes n$ systems are given in \cite{Xiong2024}. The readers are referred to \cite{Bengtsson2009, Aubrun2017} for backgrounds on the geometric studies of quantum states.
\section{Maximum purity of absolutely separable and absolutely PPT states}
\subsection{The two-qubit case}
In this subsection, we find the maximum value of the purity of absolutely separable two-qubit states. Our result confirms the numerical prediction by Slater \cite{Slater2009}.
\begin{theorem}\label{purity}
	The maximum purity of absolutely separable two-qubit states is $\dfrac{3}{8}$ and this maximum value is attained iff the  eigenvalues  $\lambda_1 = \lambda_2 = \dfrac{1}{4} + \dfrac{1}{4\sqrt{2}}$ and $\lambda_3 = \lambda_4=\dfrac{1}{4} - \dfrac{1}{4\sqrt{2}}$.
\end{theorem}
\begin{proof}
	From the inequality in (\ref{Eq2}), we deduce that $\frac{1}{2}(\lambda_1 - \lambda_3)^2\le 2\lambda_2\lambda_4$. Hence, combine with $\lambda_1 + \lambda_2 + \lambda_3 + \lambda_4 = 1$, we can bound the purity as follows.
	\begin{equation} \label{Eq4}
		\begin{aligned}
			{\rm tr}(\rho^2) &= \lambda_1^2 + \lambda_2^2 + \lambda_3^2 + \lambda_4^2\\ 
			&=\lambda_1^2 + \lambda_3^2 + (\lambda_2 + \lambda_4)^2 - 2\lambda_2\lambda_4\\
			&=\lambda_1^2 + \lambda_3^2 + (1 - \lambda_1 - \lambda_3)^2 - 2\lambda_2\lambda_4\\
			&= \left[\frac{1}{2}(\lambda_1 - \lambda_3)^2 - 2\lambda_2\lambda_4\right] + \frac{1}{6}[2 - 3(\lambda_1 + \lambda_3)]^2 + \frac{1}{3} \\
			&\le \frac{1}{6}[2 - 3(\lambda_1 + \lambda_3)]^2 + \frac{1}{3}. 
		\end{aligned}
	\end{equation}
	Since $\lambda_1 \ge \lambda_2\ge \lambda_3  \ge \lambda_4$, it follows that $2\lambda_1 + 2\lambda_3 \ge \lambda_1 + \lambda_2 + \lambda_3 + \lambda_4 = 1$. Thus, 
	\begin{equation}\label{Eq5}
		\lambda_1 + \lambda_3 \ge \frac{1}{2}.   	
	\end{equation}
	On the other hand, $1 - \lambda_1=\lambda_2 + \lambda_3 + \lambda_4 \ge 2\lambda_3 + \lambda_4 \ge 2\lambda_3$. So 
	\begin{equation}\label{Eq6}
		\lambda_3 \le \frac{1 - \lambda_1}{2}.	
	\end{equation}
	Again, by the inequality (\ref{Eq2}), we get $\lambda_1 \le \lambda_3 + 2\sqrt{\lambda_2\lambda_4} \le \lambda_3 + \lambda_2+\lambda_4 = 1 - \lambda_1$. It follows that 
	\begin{equation}\label{Eq7}
		\lambda_1 \le \frac{1}{2}.
	\end{equation}
	Using \eqref{Eq6} and \eqref{Eq7}, we get 
	\begin{equation}\label{Eq8}
		\lambda_1 + \lambda_3 \le \lambda_1 + \frac{1 - \lambda_1}{2} = \frac{1}{2} + \frac{\lambda_1}{2} \le \frac{1}{2} + \frac{1}{4} = \frac{3}{4}
	\end{equation}
	Combining \eqref{Eq5} and \eqref{Eq8}, it follows that $\lambda_1 + \lambda_3 \in \left[\dfrac{1}{2};\dfrac{3}{4}\right]$.
	
	Plugging $t = \lambda_1 + \lambda_3$ into the left-hand side of the last inequality of \eqref{Eq4},  we obtain the following function $$f(t) = \frac{1}{6}(2-3t)^2 + \frac{1}{3},\ \ t \in \left[\dfrac{1}{2};\dfrac{3}{4}\right].$$
	It is easy to see that the quadratic function $f(t),\ t \in \left[\dfrac{1}{2};\dfrac{3}{4}\right],$ attains its maximum value  of $\dfrac{3}{8}$, precisely when $t = \dfrac{1}{2}$. So the maximum value of the ${\rm tr}(\rho^2)$ is $\dfrac{3}{8}$ as required.
	
	Moreover, the maximum value attains precisely when both equalities in \eqref{Eq4} and \eqref{Eq5} hold. Thus, we derive the following system of equations  
	\begin{center}
		$\begin{cases}
			\lambda_1 &= \lambda_3+ 2 \sqrt{\lambda_2\lambda_4}\\
			\lambda_1 &= \lambda_2\\
			\lambda_3 &=\lambda_4\\
			\lambda_1 + \lambda_3 &= \frac{1}{2}
		\end{cases}$  
	\end{center}  
	The system has a unique solution:    
	$\lambda_1 = \lambda_2 = \frac{1}{4} + \frac{1}{4\sqrt{2}}$ and $\lambda_3 = \lambda_4 = \frac{1}{4} - \frac{1}{4\sqrt{2}}.$ So the theorem is proved.
	
\end{proof}
\subsection{The qubit-qudit case}
In this subsection, we present our numerical results on determining the maximum value of the purity of absolutely separable qubit-qudit states.  We used MATLAB's Optimization Toolbox, specifically the "fmincon" function, to maximize the purity function under the constraint in equation \eqref{Eq2}. These results are summarized in Table 1 below.
\begin{small}
	\begin{table}[!ht]\label{table1}
		\begin{tabular}{|l|l|l|l|l|}
			\hline
			$ \text{Dimension} $ & \text{Max. purity} &  \text{Eigenvalues}  \\ \hline

			$2 \otimes 3$ & $0.22$   & $\lambda_1 = \lambda_2 = 0.3, \lambda_3 = \lambda_4 = \lambda_5 = \lambda_6 = 0.1$\\  \hline
			
			$2 \otimes 4$ & $0.1667$  & $\lambda_1 = \lambda_2 = 0.25, \lambda_3 = \dots = \lambda_8 = 0.8333$ \\ \hline
			
			$2 \otimes 5$ & $0.1328$  & $\lambda_1 = \lambda_2 = \lambda_3= 0.1875, \lambda_4 = \dots = \lambda_{10} = 0.0625$ \\ \hline   		
			$2 \otimes 6$ & $0.11$    & $\lambda_1 =\lambda_2= \lambda_3= 0.15, \lambda_4 = \dots = \lambda_{12} = 0.05$  \\ \hline 
			
			$2 \otimes 7$ & $0.095$   & $\lambda_1 = \dots = \lambda_4 = 0.1368, \lambda_5 = \dots = \lambda_{14} = 0.0455$ \\ \hline 
			
			$2 \otimes 8$ & $0.0828$   & $\lambda_1 = \dots = \lambda_4 = 0.1154, \lambda_5 = \dots = \lambda_{16} = 0.0385$   \\ \hline 
			
			$2 \otimes 9$ & $0.0693$  &  $\lambda_1 = 0.0992, \lambda_2 = 0.0982, \lambda_3 = 0.0969, \lambda_4=0.0949$,    \\ 
			
			&  &  $\lambda_5 = 0.0918, \lambda_6 = 0.0850, \lambda_7=0.0429, \lambda_8 = 0.0397$   \\ 
			
			&  &  $\lambda_9 = 0.0381, \lambda_{10} = 0.0371, \lambda_{11}=0.0363, \lambda_{12} = 0.0357$   \\ 
			&  &  $\lambda_{13} = 0.0351, \lambda_{14} = 0.0347, \lambda_{15}=0.0343, \lambda_{16} = 0.0340$  \\ 
			&  &  $\lambda_{17} = 0.0335, \lambda_{18} = 0.0326$ \\ 
			\hline
			
			$2 \otimes 10$ & $0.0618$  &  $\lambda_1 = 0.0912, \lambda_2 = 0.0902, \lambda_3 = 0.0889, \lambda_4=0.0871$,    \\ 
			
			&  &  $\lambda_5 = 0.0842, \lambda_6 = 0.0778, \lambda_7=0.0441, \lambda_8 = 0.0396$   \\ 
			
			&  &  $\lambda_9 = 0.0374, \lambda_{10} = 0.0360, \lambda_{11}=0.0349, \lambda_{12} = 0.0341$   \\ 
			&  &  $\lambda_{13} = 0.0335, \lambda_{14} = 0.0329, \lambda_{15}=0.0324, \lambda_{16} = 0.0320$  \\ 
			&  &  $\lambda_{17} = 0.0316, \lambda_{18} = 0.0313, \lambda_{19} = 0.0308, \lambda_{20} = 0.0298$ \\
			\hline
			
		\end{tabular}
		\vskip0.5cm
		\caption{Numerical estimates of the maximum purity of absolutely separable qubit-qudit states}
	\end{table}
\end{small}

On the basis of these numerical estimates, we propose the following conjecture regarding the exact maximum purity.

\begin{conjecture} \label{conj1} The maximum purity of  absolutely separable states $\rho \in M_2\otimes M_n, n\ge 3,$ is given by
	$$\begin{cases}
		\frac{2}{3n} \ &\text{if}\ n \ \text{is even,}  \\
		
		\frac{6n+4}{(3n+1)^2} \ &\text{if}\ n \ \text{is odd.}
	\end{cases}$$
	Furthermore, this maximum value is reached when $$\begin{cases}
		\lambda_1= \cdots =\lambda_{\frac{n}{2}}=\frac{1}{n}, \lambda_{\frac{n}{2}+1}=\cdots=\lambda_{2n}=\frac{1}{3n}  \ &\text{if}\ n \ \text{is even,}  \\
		
		\lambda_1= \cdots =\lambda_{\frac{n+1}{2}}=\frac{3}{3n+1}, \lambda_{\frac{n+3}{2}}=\cdots=\lambda_{2n}=\frac{1}{3n+1} \ &\text{if}\ n \ \text{is odd.}
	\end{cases}$$ 
\end{conjecture} 
We compare the numerical values and the conjectured values in Figure 1 below.
\begin{figure}[h!]
	\centering
	\includegraphics[height=7cm,width=10cm]{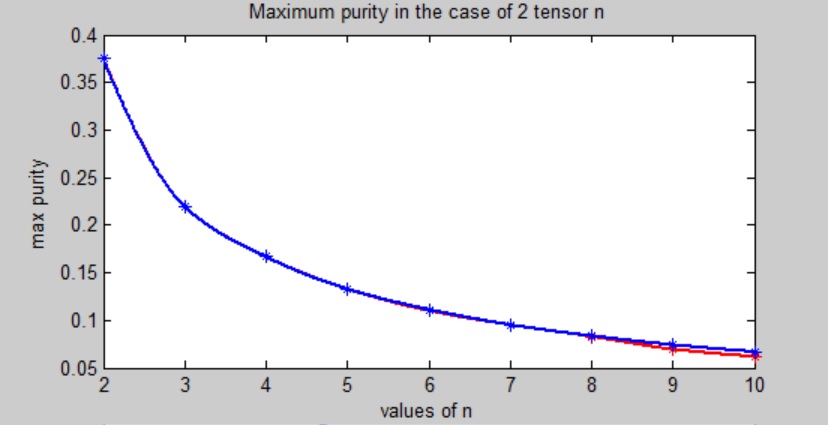}
	\caption{Plot of the maximum purity of absolutely separable states versus the dimension $n$ in a $2\otimes n$ system, where the red curve represents the numerical values and the  blue curve represents the conjectured values.}\label{Plot_ASEP2d}
\end{figure}
\subsection{The qutrit-qudit case} 
We maximize the purity function for absolutely PPT states using the constraints in \eqref{Eq3}, again utilizing the MATLAB “fmincon” function. The results are presented in Table 2.
\begin{small}
	\begin{table}\label{table2}
		\begin{tabular}{|l|l|l|l|l|}
			\hline
			Dimension & Max. purity  & Eigenvalues \\ \hline
		  $3 \otimes 2$ & $0.22$   & $\lambda_1 = \lambda_2 = 0.3, \lambda_3 = \lambda_4 = \lambda_5 = \lambda_6 = 0.1$\\  \hline
            $3 \otimes 3$ & $0.1405$   & $\lambda_1=0.2727, \lambda_2 = \dots = \lambda_{9} = 0.0909$           \\ \hline
			
			$3 \otimes 4$ & $0.0910$  &  $\lambda_1=0.1324, \lambda_2=0.1175, \lambda_3 =0.1093, \lambda_4=0.1040,$ \\ 
			
			&   &  $\lambda_5=0.0843, \lambda_6=0.0758, \lambda_7 =0.0723, \lambda_8=0.0695,$ \\
			
			&   &  $\lambda_9=0.0664, \lambda_{10}=0.0628, \lambda_{11} =0.0588, \lambda_{12}=0.0467$ \\ \hline  
			
			$3 \otimes 5$ & $0.0717$  &  $\lambda_1=0.1007, \lambda_2=0.0930, \lambda_3 =0.0887, \lambda_4=0.0856,$ \\ 
			
			&   &  $\lambda_5=0.0806, \lambda_6=0.0726, \lambda_7 =0.0652, \lambda_8=0.0607,$ \\
			
			&   &  $\lambda_9=0.0578, \lambda_{10}=0.0558, \lambda_{11} =0.0539, \lambda_{12}=0.0518$ \\ 
			
			&   &  $\lambda_{13}=0.0491, \lambda_{14}=0.0458, \lambda_{15} =0.0385$ \\ \hline
			
			$3 \otimes 6$ & $0.0625$  &  $\lambda_1=\dots = \lambda_6 = 0.0833, \lambda_7 = \dots = \lambda_{18} = 0.0417$ \\ \hline
			
			$3 \otimes 7$ & $0.0535$  &  $\lambda_1=\dots = \lambda_6 = 0.0690, \lambda_7 = \lambda_{8} = 0.0689,$ \\ 
			
			&   &  $\lambda_9 = \dots = \lambda_{21} = 0.0345$ \\\hline
			
			$3 \otimes 8$ & $0.0456$  &  $\lambda_1=0.0596, \lambda_2=0.0587, \lambda_3 =0.0583, \lambda_4=0.0580,$ \\ 
			
			&   &  $\lambda_5=0.0576, \lambda_6=0.0571, \lambda_7 =0.0565, \lambda_8=0.0557,$ \\
			
			&   &  $\lambda_9=0.0542, \lambda_{10}=0.0506, \lambda_{11} =0.0347, \lambda_{12}=0.0331,$ \\ 
			
			&   &  $\lambda_{13}=0.0323, \lambda_{14}=0.0318, \lambda_{15} =0.0314, \lambda_{16} = 0.0311,$ \\
			
			&   &  $\lambda_{17}=0.0309, \lambda_{18}=0.0307, \lambda_{19} =0.0305, \lambda_{20} = 0.0303,$ \\ 
			
			&   &  $\lambda_{21}=0.0300, \lambda_{22}=0.0296, \lambda_{23} =0.0291, \lambda_{24} = 0.0282$ \\ \hline
			
			$3 \otimes 9$ & $0.0402$  &  $\lambda_1=0.0534, \lambda_2=0.0525, \lambda_3 =0.0521, \lambda_4=0.0517,$ \\ 
			
			&   &  $\lambda_5=0.0514, \lambda_6=0.0509, \lambda_7 =0.0504, \lambda_8=0.0496,$ \\
			
			&   &  $\lambda_9=0.0485, \lambda_{10}=0.0465, \lambda_{11} =0.0406, \lambda_{12}=0.0334,$ \\ 
			
			&   &  $\lambda_{13}=0.0313, \lambda_{14}=0.0302, \lambda_{15} =0.0295, \lambda_{16} = 0.0290,$ \\
			
			&   &  $\lambda_{17}=0.0286, \lambda_{18}=0.0282, \lambda_{19} =0.0280, \lambda_{20} = 0.0277,$ \\ 
			
			&   &  $\lambda_{21}=0.0275, \lambda_{22}=0.0274, \lambda_{23} =0.0271, \lambda_{24} = 0.0269$ \\ 
			
			&   &  $\lambda_{25}=0.0265, \lambda_{26}=0.0260, \lambda_{27} =0.0252$ \\ \hline
			
			$3 \otimes 10$ & $0.0360$  &  $\lambda_1=0.0490, \lambda_2=0.0481, \lambda_3 =0.0476, \lambda_4=0.0473,$ \\ 
			
			&   &  $\lambda_5=0.0469, \lambda_6=0.0464, \lambda_7 =0.0457, \lambda_8=0.0449,$ \\
			
			&   &  $\lambda_9=0.0435, \lambda_{10}=0.0412, \lambda_{11} =0.0369, \lambda_{12}=0.0327,$ \\ 
			
			&   &  $\lambda_{13}=0.0304, \lambda_{14}=0.0291, \lambda_{15} =0.0283, \lambda_{16} = 0.0276,$ \\
			
			&   &  $\lambda_{17}=0.0271, \lambda_{18}=0.0267, \lambda_{19} =0.0264, \lambda_{20} = 0.0261,$ \\ 
			
			&   &  $\lambda_{21}=0.0258, \lambda_{22}=0.0256, \lambda_{23} =0.0254, \lambda_{24} = 0.0253$ \\ 
			
			&   &  $\lambda_{25}=0.0251, \lambda_{26}=0.0249, \lambda_{27} =0.0247, \lambda_{28}=0.0243,$ \\
			
			&   &  $\lambda_{29}=0.0239, \lambda_{30} =0.0231$ \\ \hline
			
		\end{tabular}
		\vskip0.5cm
		\caption{Numerical estimates of the maximum purity of  absolutely PPT qutrit-qudit states.}
	\end{table}
\end{small}

The numerical finding in Table 2 and Conjecture \ref{conj1} for qubit-qudit states lead us to the following conjecture for the qutrit-qudit case.
\begin{conjecture} \label{conj2} The maximum purity of  absolutely PPT states $\rho \in M_3\otimes M_n, n\ge 2,$ is given by
	$$\begin{cases}
		\frac{4}{9n} \ &\text{if}\ n \equiv 0 \mod 4,   \\
		
		\frac{36n+8}{(9n+1)^2} &\text{if}\ n \equiv 1 \mod 4 , \\
		\frac{36n+16}{(9n+2)^2} &\text{if}\ n \equiv 2 \mod 4 , \\
		\frac{36n-8}{(9n-1)^2}  &\text{if}\ n \equiv 3 \mod 4.
	\end{cases}$$
	Moreover, the maximum value is attained when $$\begin{cases}
		\lambda_1= \cdots =\lambda_{\frac{3n}{4}}=\frac{2}{3n}, \lambda_{\frac{3n}{4}+1}=\cdots=\lambda_{3n}=\frac{2}{9n}  \ &\text{if}\ n \equiv 0 \mod 4,   \\
		
		\lambda_1= \cdots =\lambda_{\frac{3n+1}{4}}=\frac{6}{9n+1}, \lambda_{\frac{3n+5}{4}}=\cdots=\lambda_{3n}=\frac{2}{9n+1} \ &\text{if}\ n \equiv 1 \mod 4, \\
		\lambda_1= \cdots =\lambda_{\frac{3n+2}{4}}=\frac{6}{9n+2}, \lambda_{\frac{3n+6}{4}}=\cdots=\lambda_{3n}=\frac{2}{9n+2} \ &\text{if}\ n \equiv 2 \mod 4, \\
		\lambda_1= \cdots =\lambda_{\frac{3n-1}{4}}=\frac{6}{9n-1}, \lambda_{\frac{3n+3}{4}}=\cdots=\lambda_{3n}=\frac{2}{9n-1} \ &\text{if}\ n \equiv 3 \mod 4.
		
	\end{cases}$$ 
\end{conjecture}
A comparison of the numerical and conjectured values is shown in Figure 2. Due to the increasing dimensionality and complexity of the constraints, the numerical estimates are noticeably smaller than the conjectured values.  

\begin{figure}[h!]
	\centering
	\includegraphics[height=6cm,width=10cm]{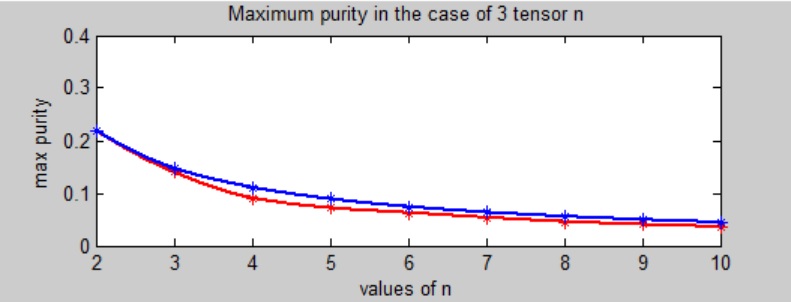}
	\caption{Plot of the maximum purity of absolutely PPT states versus the dimension $n$ in a $3\otimes n$ system, where the red curve represents the numerical values and the  blue curve represents the conjectured values.}\label{Plot_ASEP3d_compare}
\end{figure}
\begin{remark} From the geometric perspective, Conjecture \ref{conj1} and \ref{conj2} can  be reformulated as follows.
	
	\begin{conjecture}  The Euclidean radius of the smallest ball in $M_2 \otimes M_n$ with Hilbert-Schmidt norm and around the maximally mixed  state that encompasses
		all absolutely separable states in a $2\otimes n$ system is given by 
		$$\begin{cases}
			\frac{1}{\sqrt{6n}} \ &\text{if}\ n \ \text{is even,}  \\
			
			\frac{\sqrt{3n^2+2n-1}}{(3n+1)\sqrt{2n}} \ &\text{if}\ n \ \text{is odd.}
		\end{cases}$$
	\end{conjecture}
	
	\begin{conjecture}  The Euclidean radius of the smallest ball in $M_3 \otimes M_n$ with Hilbert-Schmidt norm and around the maximally mixed state  that encompasses all absolutely PPT states in a $3\otimes n$ system  is given by 
		$$\begin{cases}
			\frac{1}{3\sqrt{n}} \ &\text{if}\ n \equiv 0 \mod 4,   \\
			
			\frac{\sqrt{27n^2+6n-1}}{(9n+1)\sqrt{3n}} &\text{if}\ n \equiv 1 \mod 4 , \\
			\frac{\sqrt{27n^2+12n-4}}{(9n+2)\sqrt{3n}} &\text{if}\ n \equiv 2 \mod 4 , \\
			\frac{\sqrt{27n^2-6n-1}}{(9n-1)\sqrt{3n}} &\text{if}\ n \equiv 3 \mod 4.
		\end{cases}$$
	\end{conjecture}
	
\end{remark}
\section{Conclusions}
In conclusion, we studied the maximum purity of absolutely separable and absolutely PPT states. We provided exact results for two-qubit states and numerical estimates for higher-dimensional systems. Additionally, we proposed conjectures for the maximum purity of absolutely separable qubit-qudit states and absolutely PPT qutrit-qudit states. Our findings can be interpreted geometrically as the determination of the Euclidean radius of the smallest ball encompassing the set of absolutely separable and absolutely PPT states. We conjecture that these sets roughly lie within a ball of Euclidean radius $O(\frac1{\sqrt{n}})$ centered at the maximally mixed state. These results offer insights into the geometric properties of absolutely separable and absolutely PPT states, and their position in the overall state space.

\section*{Acknowledgments}

The second-named author is supported by Vietnam Academy of Science and Technology, grant number NVCC01.12/24-25.

\end{document}